\documentclass[a4paper,UKenglish,cleveref, autoref, thm-restate]{lipics-v2021}

\usepackage{macros}

\title{Tight bounds for clique-packing parameterized by clique-width}  
 
\author{Narek Bojikian}{Humboldt-Universität zu Berlin, Germany}{bojikian@hu-berlin.de}{https://orcid.org/0000-0003-1072-4873}{}

\author{Stefan Kratsch}{Humboldt-Universität zu Berlin, Germany}{stefan.kratsch@hu-berlin.de}{https://orcid.org/0000-0002-0193-7239}{}

\authorrunning{N.\ Bojikian and S.\ Kratsch}
  
\Copyright{Narek Bojikian and Stefan Kratsch}  

\category{}  
\relatedversion{}

 \EventEditors{John Q. Open and Joan R. Access}
\EventNoEds{2}
\EventLongTitle{42nd Conference on Very Important Topics (CVIT 2016)}
\EventShortTitle{CVIT 2016}
\EventAcronym{CVIT}
\EventYear{2016}
\EventDate{December 24--27, 2016}
\EventLocation{Little Whinging, United Kingdom}
\EventLogo{}
\SeriesVolume{42}
\ArticleNo{23}

\pdfoutput=1 
\hideLIPIcs
\nolinenumbers

\begin{document}
\maketitle

\begin{abstract}
    In the $d$-\ClqPackp problem, given a graph $G$ and an integer $\target$, we need to decide whether $G$ contains a set of $t$ pairwise vertex-disjoint cliques of size $d$ each. This generalizes \textsc{Triangle Packing} and it is NP-complete for all $d\geq 3$. For each such $d$, we show how to solve the problem in $n^{\Oh(k^{d-1})}$ time where $k$ is the clique-width of the graph (with a $k$-expression of $G$ given in the input). We complement this by showing that, assuming the Exponential-Time Hypothesis (ETH), there is no algorithm that solves the problem in $n^{o(k^{d-1})}$ time for any fixed $d\geq 3$, already for the special case of seeking a partition into cliques of size $d$. Our proof also entails $W[1]$-hardness of $d$-\ClqPackp (and $d$-\textsc{Clique Partition}) parameterized by clique-width for each $d\geq 3$. Our work continues a series of results on ETH-tight bounds for fundamental graph problems started by Fomin et al.\ (SICOMP 2010+2014) who obtained tight bounds for \textsc{Max-Cut} and \textsc{Edge Dominating Set}.
\end{abstract}

\section{Introduction}

Parameterized complexity studies the influence of instance properties such as solution size or structural restrictions on problem complexity. To this end, such properties are quantified as parameters and the main goal is to find parameterized algorithms that are (reasonably) fast when said parameters are small. By now, there is an extensive toolbox of algorithmic techniques as well as a variety of structural parameters for dealing with problems that are NP-hard in general. These result in FPT-algorithms (fixed-parameter tractable) with running time $f(k)\cdot n^{\Oh(1)}$ and in (slower) XP-algorithms with running time $n^{f(k)}$, where $n$ is the input size, $k$ is the parameter, and $f$ is any computable function. This is complemented by a well-developed understanding of lower bounds, proving conditional optimality of the dependence on $k$ subject to, e.g., ETH or SETH\footnote{The Exponential-Time Hypothesis (ETH) posits that there is $c>1$ such that no algorithm solves \textsc{$3$-SAT} in $\Oh(c^n)$ time. The Strong Exponential-Time Hypothesis (SETH) posits that for each $c<2$ there is $q\in\mathbb{N}$ such that no algorithm solves \textsc{$q$-SAT} in $\Oh(c^n)$ time.}, or conditionally ruling out FPT- or even XP-algorithms, e.g., conditioned on $P\neq NP$ or $FPT\neq W[1]$.\footnote{The class FPT contains all parameterized problems that have an FPT-algorithm relative to their parameter. The class $W[1]$ contains all problems reducible to \textsc{Clique} parameterized by solution size.}

In this work, we are interested in (conditionally) tight complexity relative to the graph parameter \emph{clique-width}. While clique-width is not as ubiquitous as treewidth, for which there is a wealth of conditionally tight upper and lower bounds, it can also be small for dense graphs.\footnote{If treewidth is bounded then so is clique-width, but clique-width may be exponential in the treewidth. Conversely, there are families of graphs of bounded clique-width but unbounded treewidth. For a graph to have treewidth at most $k$, it cannot have more than $kn$ edges.} Recently, a number of works have given tight bounds under SETH for many fundamental graph problems parameterized by clique-width, most taking the form that an $\alpha^k n^{\Oh(1)}$ time (FPT-)algorithm is known/shown and any $(\alpha-\varepsilon)^k n^{\Oh(1)}$ time algorithm would refute SETH; here $\alpha$ is problem-specific~\cite{DBLP:journals/siamdm/Lampis20,DBLP:conf/esa/HegerfeldK23,DBLP:conf/icalp/BojikianK24,DBLP:journals/talg/GanianHKOS24,DBLP:journals/mst/JaffkeLL24,DBLP:conf/iwpec/BojikianK25,DBLP:journals/corr/abs-2512-01900/BojikianK25}. Earlier, starting with Fomin et al.~\cite{DBLP:journals/siamcomp/FominGLS10,DBLP:journals/siamcomp/FominGLS14}, tight upper and (conditional) lower bounds were obtained, usually with $n^{f(k)}$ upper bound and ruling out $n^{o(f(k))}$ subject to ETH~\cite{DBLP:journals/siamcomp/FominGLS14,DBLP:journals/tcs/BergougnouxK19,DBLP:journals/talg/FominGLSZ19,DBLP:journals/algorithmica/BergougnouxKK20}. In particular, for \textsc{Edge Dominating Set}, \textsc{Max Cut}, and \textsc{Hamiltonian Cycle} there are algorithms running in $n^{\Oh(k)}$ time when given a $k$-expression of the graph, while $n^{o(k)}$ time is ruled out under ETH~\cite{DBLP:journals/siamcomp/FominGLS14,DBLP:journals/talg/FominGLSZ19,DBLP:journals/algorithmica/BergougnouxKK20}. For \textsc{Graph Coloring}, however, there is a $n^{\Oh(2^k)}$ time algorithm by Kobler and Rotics~\cite{DBLP:journals/dam/KoblerR03} and Fomin et al.~\cite{DBLP:journals/talg/FominGLSZ19} ruled out $n^{2^{o(k)}}$ time under ETH.\footnote{This would leave room for, e.g., a $n^{1.5^k}$ time algorithm but double exponential dependence is necessary.}

The unusual dependence of $2^k$ in the exponent of $n$ for \textsc{Graph Coloring} compared to $\Oh(k)$ for the other problems naturally leads to the question of what other (conditionally) tight complexities are open. Presumably, there are artificial problems for any dependence on $k$, but are there also natural ones with e.g.\ a polynomial of $k$ in the exponent? Moreover, outside of the many problems that are MSO expressible and thereby FPT due to a result of Courcelle~\cite{DBLP:journals/mst/CourcelleMR00} our knowledge about the (conditionally) optimality of XP-algorithms for problems parameterized by clique-width is still quite limited. 

With this motivation, we study the $d$-\ClqPackp problem where, given a graph $G$ and an integer $t$, the task is to decide whether there exists a family of at least $\target$ vertex-disjoint cliques of size exactly $d$ each in $G$. We also consider the special case of $d$-\ClqPartp where we seek a partition into cliques of size $d$ each. These are natural generalizations of \textsc{Triangle Packing} and \textsc{Partition into Triangles} and they are NP-complete for each $d\geq 3$. They are also natural stepping stones for considering more general (induced) subgraph packing problems.

\subparagraph{Our results.}

We provide tight bounds for the complexity of $d$-\ClqPackp and $d$-\ClqPartp relative to clique-width for each value $d\geq 3$ (modulo ETH). Our reductions also directly show the $W[1]$-hardness of these problems parameterized by clique-width.
Since $d$-\ClqPackp clearly solvable in polynomial time for $d\leq 2$ (the case $d=2$ is equivalent to the maximum matching), this results in a full classification of the parameterized complexity of these problems with respect to clique-width for each value of $d$. We divide our work into two parts mainly proving the following two theorems.

\begin{theorem}\label{theo:ub}
    For each positive integer $d\geq 3$, the $d$-\ClqPackp problem can be solved in time $n^{\Oh(k^{d-1})}$ when a $k$-clique-expression of the input graph is given. 
\end{theorem}

In the latter part, we show that for each value $d\geq 3$ even the $d$-\ClqPartp problem cannot be solved in time $n^{o(k^{d-1})}$ under ETH, even when the input graph is provided with a linear $k$-expression. 

\begin{theorem}\label{theo:lb}
    For any integer $d\geq 3$, the $d$-\ClqPartp problem is $W[1]$-hard, and cannot be solved in time $n^{o(k^{d-1})}$ assuming ETH, even when the input graph is given together with a linear $k$-expression.
\end{theorem}

\subparagraph{Related work.}
Fürer proposed two generalizations of clique-width, called fusion width~\cite{DBLP:conf/latin/Furer14} and multi-clique-width~\cite{DBLP:conf/innovations/Furer17}, that in particular avoid the exponential blow-up compared to treewidth. Chekan and Kratsch~\cite{DBLP:conf/mfcs/ChekanK23} showed that several tight bounds carry over from clique-width to both new variants. 

Bergougnoux et al.~\cite{DBLP:conf/stacs/BergougnouxKN23} showed that assuming ETH there is no $2^{o(k^2)} \cdot n^{\Oh(1)}$ time algorithm for \textsc{Independent Set} parameterized by rank-width, an even more general parameter. It is known that rank-width is upper-bounded by clique-width, but \textsc{Independent Set} parameterized by clique-width can be solved in $2^k\cdot n^{\Oh(1)}$ time, which is optimal under SETH~\cite{DBLP:conf/esa/IwataY15,DBLP:journals/talg/LokshtanovMS18}.

\subparagraph{Organization.}
In Section~\ref{section:preliminaries} we give a formal definition of clique-width along with a useful variant called \emph{NLC-width}. Section~\ref{section:upperbound} presents our algorithm and proves Theorem~\ref{theo:ub}. In Section~\ref{section:lowerbound} we prove Theorem~\ref{theo:lb} by giving a reduction from \textsc{Multi-Colored Clique}. We conclude in Section~\ref{section:conclusion}.

\section{Preliminaries}\label{section:preliminaries}

For a positive integer $d$, a \emph{$d$-clique} in a graph $G$ is a subset of vertices of $G$ of size $d$ that induces a complete subgraph. A \emph{$d$-clique packing} in $G$ is a family of vertex-disjoint $d$-cliques in $G$.
We define the $d$-\ClqPackp problem as follows:

\problemdef{\ClqPackp}{
A graph $G=(V,E)$ and an integer $t\in \mathbb{N}$.}{Is there a $d$-clique packing of size at least $t$ in $G$?}

The $d$-\ClqPartp problem is defined as the special case with $|V| = td$.

A \emph{reduction} $f$ between two parameterized problems $A$ and $B$ is a mapping that maps each instance $(I_A, k_A)$ of $A$ to an instance $(I_B, k_B)$ of $B$ such that $f$ is computable, and $(I_A, k_A) \in A$ if and only if $(I_B, k_B) \in B$. We say that $f$ is a \emph{parameterized reduction} if it can be computed in FPT time and it holds that the size of the output parameter is bounded by some computable function of the input parameter. It is well-known that the class $W[1]$ is closed under parameterized reductions. We refer to \cite[Chapter 13]{DBLP:books/sp/CyganFKLMPPS15} for details.

\subparagraph{Clique Expression}

A \emph{labeled graph} is a graph $G=(V, E)$ together with a \emph{labeling function} $\lab\colon V\rightarrow \mathbb{N}$. We usually omit the function $\lab$ and assume that it is implicitly given with $G$. We say that $G$ is \emph{$k$-labeled}, if it holds that $\lab(v)\leq k$ for all $v\in V$. A \emph{labeled forest} is a labeled graph that is a forest.
We define a \emph{clique expression} $\mu$ as a well-formed expression defined by the following operations on labeled graphs:
\begin{itemize}
    \item \emph{Introduce vertex} $i(v)$ for $i\in\mathbb{N}^+$. This operation constructs a graph containing a single vertex $v$ labeled $i$.
    \item The \emph{relabel} operation $\clqrel_{i,j}(G)$ for $i,j \in \Nats^+$, $i\neq j$. This operation changes the labels of all vertices in $G$ labeled $i$ to the label $j$.
    \item The \emph{join} operation $\clqjoin_{i,j} (G)$ for $i,j\in\Nats^+, i\neq j$. The constructed graph results from $G$ by adding all edges between the vertices labeled $i$ and the vertices labeled $j$.
    \item The \emph{union} operation $G_1 \clqunion G_2$. The resulting graph is the disjoint union of $G_1$ and $G_2$.
\end{itemize}

We denote the graph resulting from a clique expression $\mu$ by $G_{\mu}$, and the constructed labeling function by $\lab_{\mu}$. We associate with a clique expression $\mu$ a syntax tree $\syntaxtree_{\mu}$ (we omit $\mu$ when clear from context) in the natural way, and associate with each node $x\in V(\syntaxtree)$ the corresponding operation. For $x\in V(\syntaxtree)$, the subtree $\syntaxtree_x$ rooted at $x$ induces a subexpression $\mu_x$. We define $G_x = G_{\mu_x}$, $V_x = V(G_x)$, $E_x=E(G_x)$ and $\lab_x = \lab_{\mu_x}$.

We say that a clique expression $\mu$ is a \emph{$k$-expression} if $G_x$ is a $k$-labeled graph for all $x\in V(\syntaxtree)$. We define the clique-width of a graph $G$ (denoted by $\cw(G)$) as the smallest value $k$ such that there exists a $k$-expression $\mu$ with $G_{\mu}$ isomorphic to $G$.
A \emph{linear $k$-expression} $\mu$ is a $k$-expression, such that $\syntaxtree_{\mu}$ is a caterpillar, i.e.\ each union node has at least one child that corresponds to an introduce vertex operation.

\subparagraph{NLC-Width.}
A closely related parameter to clique-width is NLC-width, where an NLC-expression is defined solely by the two following operations:
The \emph{introduce operation} $i(v)$, resulting in a graph consisting of a single vertex $v$ labeled $i$. The second operation is called the \emph{join operation} $\oplus_\alpha^\beta$ for an \emph{edge mapping} $\alpha\subseteq [k]^2$ and a \emph{relabeling function} $\beta: [k]\rightarrow [k]$.
Given two $k$-graphs $G_1$ and $G_2$, $G_1 \oplus_\alpha^\beta G_2$ is the $k$-graph obtained by performing the following steps in order:
\begin{enumerate}
\item First, we construct the graph $G^1$, the disjoin union of $G_1$ and $G_2$.
\item After that we add to $G^1$ all edges between vertices $v_1$ of $G_1$ with label $i_1$ and vertices $v_2$ of $G_2$ of label $i_2$ such that $(i_1,i_2)\in\alpha$ resulting in the graph $G^2$.
\item Finally, for each vertex $v$ of $G^2$, we set the label of $v$ to $\beta(\lab(v))$ resulting in the graph $G_1 \oplus_\alpha^\beta G_2$.
\end{enumerate}

The \emph{NLC-width} of a graph $G$ (denoted $\nlcw(G)$) is the smallest value $k$ such that there exists an $k$-NLC-expression generating a graph isomorphic to $G$. 
We accompany each NLC-expression with a syntax tree $\syntaxtree$ in a natural way, where each leaf corresponds to an introduce operation and each internal node corresponds to a join operation. We  denote the nodes of $\syntaxtree$ by $\nodes=V(\syntaxtree)$. For $x\in\nodes$, we denote by $\syntaxtree_x$ the subtree of $\syntaxtree$ rooted at $x$, and by $G_x=(V_x, E_x)$ the graph obtained by the subexpression corresponding to the $\syntaxtree_x$.

It is well-known~\cite{johansson1998clique,BodlaenderH12} that it holds for every graph $G$ that
\[\nlcw(G)\leq \cw(G)\leq 2\cdot \nlcw(G).\]
Moreover, this relationship is constructive and so provided a $k$-clique-expression of a graph $G$ of width $k$, a $k$-NLC-expression of $G$ can be computed in polynomial time.

\section{Upper Bound}\label{section:upperbound}

We build our algorithm over NLC-expressions rather than clique-expressions, since the former allow for a simpler presentation. We can do this without loss of generality, since an NLC-expression of width at most $k$ can be computed in polynomial time from a given $k$-clique-expression.
Along this upper bound, let us fix a constant integer $d\geq 3$, and let $(G=(V,E),\target)$ be the given instance of the $d$-clique packing problem, where $G$ is given together with an NLC-expression of width $k$. Let $\syntaxtree$ be the syntax tree of the given NLC-expression, and let $\nodes \coloneq V(\syntaxtree)$ be its nodes.

Intuitively, a partial solution at a node $x$ is defined as a partition of the vertices in $G_x$ into cliques of size \emph{at most} $d$, where a solution is a partition at the root node with at least $\target$ cliques of 
size \emph{exactly} $d$.
We will define fingerprints corresponding to different partial solutions, that will allow us to count solutions in the claimed running time. Crucially, a clique can be represented by the set of its labels, where we have $\Oh(k^d)$ different such sets. A fingerprint of a partial solution then counts the number of times each such set of labels appears as a clique in this partial solution.

At first glance, this implies an algorithm with running time at least $n^{\Omega(k^d)}$ since this is the number of different fingerprints we keep track of. However, our main observation in the upper bound is that, while keeping the sets of labels in a fingerprint is essential to extend cliques correctly along the expression---as it allows us to combine two smaller cliques in a partial solution correctly---once a clique hits the target size $d$, we can ``forget'' its set of labels, and keep an aggregate counter for all such cliques independent of their labels. This reduces the running time to $n^{\Oh(k^{d-1})}$, since each finger print can now be defined as a mapping from $\Oh(k^{d-1})$ different sets of labels to some integer bounded by $n$, with a single added integer that represents the total count of $d$-cliques. In the following, we formalize this idea and present the algorithm.
We start with some notation.

\begin{definition}
Given a vertex set $S$ of a labeled graph $G$, we define the \emph{type} of $S$, denoted $\lab^M(S)$, as the multiset of labels of the vertices in $S$.
\end{definition}

Let $\mathcal{M}_d$ be the family of all multisets of size at most $d-1$ over the set $[k]$.
We define the family of states $\states \coloneq [n]_0^{\mathcal{M}_d}$ of all mappings that assign an integer in $[n]_0$ to each multiset in $\mathcal{M}_d$. We use this family to index the tables of our dynamic programming algorithm.

\begin{definition}
    Given a labeled graph $G$, a \emph{partial solution} in $G$ is a partitioning of the vertices in the subgraph $G_x$ into cliques of size at most $d$. Given a partial solution $P$, we define the \emph{state} of this partitioning as the mapping $\statef(P)\in\states$ where for each $M\in \mathcal{M}_d$, $\statef(M)$ is the number of cliques $C$ of $P$ with $\lab^M(C)=M$. We define the \emph{weight} of the partitioning as the number of cliques in $P$ of size exactly $d$.
\end{definition}

Now we define the tables $A^t_x$ which indicate the existence of a partial solution of a specific state. We show that our dynamic programming routine computes these tables exactly.

\begin{definition}
We define the tables $A^t_x \in \{0,1\}^{\states}$ for each node $x \in \nodes$ and value $t \in [n]_0$, with $A^{t}_x[s]=1$ if and only if there exists a partial solution in $G_x$ of weight $t$ and state $s$.
\end{definition}

The following observation holds since in a solution $S$, any vertex of $G$ not covered by $S$ can be covered by a clique of size one in a partition $P$.

\begin{observation}\label{obs:ub-A-sol}
    The graph $G$ admits a $d$-clique packing of size $\target$ if and only if there exists a state $s\in \states$ such that $A_r^{\target}[s] = 1$.
\end{observation}

Now we aim to describe a dynamic programming routine that computes exactly the tables $A^t_x$ for each node $x\in \nodes$ and each value $t\in [n]_0$.
When handling a join node, it is more convenient to assign distinct labels to the vertices of the two unified graphs. This extends our definition of states.
Let $\mathcal{M}^{2k}_d$ be the family of all multisets of size at most $d-1$ over the set $[2k]$. Let $\states^{2k}$ be the family of all extensions of the mappings in $\states$ to $\mathcal{M}^{2k}_d$. 

Before we present the algorithm, we define the following operation that describes how a state of a partial solution changes after relabeling the vertices in an NLC-expression. We also use this operation to shift the labels in a state corresponding to the right child of a join node before combining it with the state corresponding to the left child, achieving the distinction of labels mentioned above.

\begin{definition}
    Given a multiset $M\in\mathcal{M}^{2k}$ and two integers $i,j\in [2k]$, we define the multiset $M_{i\rightarrow j}$ obtained from $M$ by replacing each occurrence of $i$ in $M$ with $j$, i.e.
    \[
    \#_{\ell}(M_{i\rightarrow j}) = \begin{cases}
    0 & ;\ell = i,\\
    \#_{j}(M) + \#_{i}(M) & ;\ell = j,\\
    \#_{\ell}(M) & ;\text{otherwise.}
    \end{cases}
    \]

    Given a mapping $s\in \states^{2k}$, we define the mapping $s_{i\rightarrow j}$ as 
    $s_{i\rightarrow j}(M) = \sumstack{M'\in\mathcal{M}_d\\(M')_ {i\rightarrow j} = M} s(M')$.
    Finally, for some integer $k'$, a relabeling function $\gamma\colon [k']\rightarrow [k']$ and a multiset $M\in \mathcal{M}^{k'}_d$, we define the multiset $M_{\rightarrow \gamma} \in \mathcal{M}^{k'}_d$ by replacing each occurrence of a label $\ell$ in $M$ with $\gamma(\ell)$. We define the mapping $s_{\rightarrow \gamma}$ as
    \[
        s_{\rightarrow \gamma}(M) = \sumstack{M'\in \mathcal{M}^{2k}_d\\ (M')_{\rightarrow \gamma}=M} s(M').
    \]
\end{definition}
 
For $s\in\states$, let $s\downop$ be the extension of $s$ to $\mathcal{M}^{2k}_d$ with $s\downop(M) = 0$ for each $M\in \mathcal{M}^{2k}_d\setminus \mathcal{M}_d$. Let $s\upop = (s\downop)_{\rightarrow \gamma}$ for the relabeling $\gamma\colon [2k]\rightarrow [2k]$ defined by $\gamma(i)=i+k$ for $i\in [k]$ and $\gamma(i)=i$ for $i\in \{k+1,\dots,2k\}$. Intuitively, $s\downop$ is the extension of $s$ that preserves the labels, whereas $s\upop$ is the extension of $s$ that shifts all labels by $k$.

\begin{definition}
    Given two mappings $s_1, s_2 \in \states$, we define the \emph{shifted union} of $s_1$ and $s_2$, where for $M\in \mathcal{M}^{2k}$ it holds that
    \[
    (s_1\punion s_2)(M) = s_1\downop(M) + s_2\upop(M).
    \]
\end{definition}

Note that $s_1\punion s_2(M) = 0$ whenever $M$ contains an element in $[k]$ and an element in $[2k]\setminus [k]$. Intuitively, this operation combines two states by shifting the labels of the latter by $k$. This allows unified handling of join nodes in a step-by-step manner.

Finally, let $(X_1, Y_1) \dots (X_h, Y_h)$ be an enumeration of all sets in $\mathcal{M}^{2k} \times \mathcal{M}^{2k}$ in an arbitrary ordering. The following notion of feasibility indicates when two cliques in a partial solution can be merged by a join operation.

\begin{definition}
Given a node $x\in \nodes$ corresponding to an operation $\oplus_{\alpha}^{\beta}$, and an integer $i\in[h]$, we call the pair $(X_i, Y_i)$ \emph{feasible} at $x$, if it holds that $|X_i\cup Y_i|\leq d$ and for each label $\ell\in X_i$ and each label $\ell'\in Y_i$, it holds that $(\ell,\ell')\in\alpha$. 
\end{definition}

\begin{algorithm}\label{ub-alg}
We define the tables $T^{\omega}_x \in \{0,1\}^{\states}$ for each node $x \in \nodes$ and each value $\omega \in [n]_0$ recursively over $\syntaxtree$ where for an introduce vertex node $\mu_x = i(v)$: we define $T^{\omega}_x[s]=1$ 
for the mapping $s$ with $s(\langle i\rangle)=1$ and $s(M)=0$ for all other multisets $M\in \mathcal{M}_d$, and $T^{\omega}_x[s']=0$ for all other mappings $s' \in \states$ different from $s$.

For a join node $x$ corresponding to the operation $\oplus_\alpha^\beta$, we start by defining the auxiliary tables $T_i^{\omega}\in \{0,1\}^{\states^{2k}}$ for each value $i\in[h]_0$ and each $\omega\in [n]_0$.
First, we define the tables  $T_0^{\omega}$ as follows:
\[
T_0^{\omega}[s] = \orstack{\omega_1+\omega_2=\omega}\orstack{s_1, s_2\in \states\\ s_1 \punion s_2 = s} T_{x_1}^{\omega_1}[s1] \land T_{x_2}^{\omega_2}[s_2],
\]
where $T_0^{\omega}[s]=1$ if there exists two values $\omega_1, \omega_2\in [n]_0$ with $\omega_1+\omega_2=\omega$ and two states $s_1, s_2\in \states$ with $s_1\punion s_2 = s$ such that $T_{x_1}^{\omega_1}[s_1] = 1$ and $T_{x_2}^{\omega_2}[s_2] = 1$, and $T_0^{\omega}[s]=0$ otherwise.

For $i\in[h]$, we compute  the tables $T^{\omega}_i$ for all $\omega\in[n]_0$ as follows:
If $(X_i, Y_i)$ is not feasible at $x$, then we set $T_i = T_{i-1}$. 
Otherwise, we start by setting $T^{\omega}_i[s] = 0$ for all values of $\omega$ and all states $s\in\states^{2k}$. Then for each state $s\in \states^{2k}$ let $r = \min\{s(X_i), s(Y_i)\}$. For each value $r'\in[r]_0$, let $s'$ be the state resulting from $s$ by subtracting $r'$ from $s(X_i)$ and $s(Y_i)$, and adding $r'$ to $s(X_i\cup Y_i)$ if $|X_i\cup Y_i|\leq d-1$. 
Then we set $T^{\omega}_i[s'] = T^{\omega}_i[s'] \lor T^{\omega}_{i-1}[s']$ if $|X_i \cup Y_i| \leq d_1$, or we set 
$T^{\omega+r'}_i[s'] = T^{\omega}_i[s'] \lor T^{\omega}_{i-1}[s']$ otherwise.

Let $\gamma^R\colon [2k]\rightarrow [2k]$ be the relabeling function defined by $\gamma^R(i)=\gamma^R(i+k)=i$ for each $i\in [k]$.
we define the auxiliary table $\hat{T}^{\omega}\in \{0,1\}^{\states^{2k}}$ where for $s\in\states^{2k}$ we define
\[
\hat{T}^{\omega}[s] = \orstack{s'\in\states^{2k}\\ (s')_{\rightarrow \gamma^R} = s} T^{\omega}_h[s'].
\]
Let $\hat{T}_0^{\omega}$ be the restriction of $\hat{T}^{\omega}$ to $\states$ with $\hat{T}_0^{\omega}[s] = \hat{T}^{\omega}[s\downop]$ for each $s\in \states$.
We define 
\[
T_x^{\omega}[s] = \orstack{s'\in\states\\ (s')_{\rightarrow \beta} = s} \hat{T}_0^{\omega}[s'].
\]
\end{algorithm}

Now we prove the correctness of the algorithm,
where we show by induction over $\syntaxtree$, that $T^{\omega}_x = A^{\omega}_x$ for each node $x\in \nodes$ and each value $\omega \in [n]_0$. In order to show this, we define the tables $\hat{A}^{\omega}$, and $\hat{A}_i^{\omega}$ for each value $i$ corresponding to the tables $\hat{T}^{\omega}$ and $T_i^{\omega}$ respectively, and show their equality, again by induction over $i$.
We get the following lemma.

\begin{lemma}\label{lem:ub-correct}
    It holds for each node $x\in \nodes$ and each value $\omega \in [n]_0$ that $T^{\omega}_x = A^{\omega}_x$.
\end{lemma}

\begin{proof}
    We prove the claim by a bottom up induction over $\syntaxtree$.

    For the base case, let $x$ be an introduce node with $\mu_x=i(v)$. Then $G_x$ consists only of the vertex $v$ with label $i$. Hence there is exactly one feasible partition of $V(G_x)$ into cliques of size at most $d$, namely the partition containing the single clique $\{v\}$. Therefore it holds that $A_x^0[s]=1$ for the unique state $s\in \states$ with $s(\langle i\rangle)=1$ and $s(M)=0$ for all other $M\in \mathcal{M}_d$, and $A_x^{\omega}[s']=0$ for all other pairs $(\omega,s')$. This matches the definition of $T_x^{\omega}$ in the algorithm, and thus $T_x^{\omega}=A_x^{\omega}$ for all $\omega\in [n]_0$.

    For the induction step, let $x$ be a join node of $\syntaxtree$, and assume that the statement holds for the children of $x$. Let $x_1$ and $x_2$ be the children of $x$, and let $(X_1,Y_1),\dots,(X_h,Y_h)$ be the ordered list of pairs of multisets used by the algorithm at node $x$.

    For each $i\in [h]_0$ and each $\omega\in [n]_0$, we define an auxiliary table $A_i^{\omega}\in \{0,1\}^{\states^{2k}}$ as follows: Let $G^i$ be the graph obtained from the shifted disjoint union of $G_{x_1}$ and $G_{x_2}$ by adding all edges between vertices of $G_{x_1}$ with labels in $X_j$ and vertices of $G_{x_2}$ with labels in $Y_j$ for each index $j\in [i]$. Then for $s\in \states^{2k}$ and $\omega\in [n]_0$, we set $A_i^{\omega}[s]=1$ if and only if there exists a partial solution in $G^i$ of weight $\omega$ and state $s$.

    We now prove by induction on $i\in [h]_0$ that $T_i^{\omega}=A_i^{\omega}$ for every $\omega\in [n]_0$.
    For $i=0$, no clique is allowed to use vertices from both sides. Hence a partition counted by $A_0^{\omega}$ is exactly the disjoint union of a partition of $G_{x_1}$ and a partition of $G_{x_2}$, where the labels on the second side are shifted by $k$. Thus $A_0^{\omega}$ is obtained by combining the tables of the two children via the shifted union operation $\punion$. By the induction hypothesis on the syntax tree, the child tables satisfy $T_{x_1}^{\omega}=A_{x_1}^{\omega}$ and $T_{x_2}^{\omega}=A_{x_2}^{\omega}$. It follows that $T_0^{\omega}=A_0^{\omega}$ for every $\omega\in [n]_0$.

    Assume now that $i\in [h]$ and that $T_{i-1}^{\omega}=A_{i-1}^{\omega}$ holds for all $\omega\in [n]_0$. We show that $T_i^{\omega}=A_i^{\omega}$ for all $\omega\in [n]_0$.
    First suppose that the pair $(X_i,Y_i)$ cannot form a clique across the join, either because $|X_i\cup Y_i|\geq d+1$ or because some label occurring in $X_i$ is not joined to some label occurring in $Y_i$. Then processing the $i$-th pair creates no new cliques, so the family of partial solutions does not change: $A_i^{\omega}=A_{i-1}^{\omega}$ for all $\omega$. Since it holds in this case, that $T_i^{\omega}=T_{i-1}^{\omega}$, the claim follows from the induction hypothesis.

    Suppose therefore that $(X_i,Y_i)$ is feasible.
    If it holds that $T_i^{\omega}[s]=1$, then, by the update rule of the algorithm, this value arises from some state $s'\in \states^{2k}$ and some integer $r'\in [\min\{s'(X_i),s'(Y_i)\}]_0$ with $T_{i-1}^{\omega'}[s']=1$, where $\omega'=\omega$ if $|X_i\cup Y_i|\leq d-1$ and $\omega'=\omega-r'$ if $|X_i\cup Y_i|=d$. By the induction hypothesis for $i-1$, there exists a packing counted by $A_{i-1}^{\omega'}[s']$. In this packing we choose $r'$ cliques of type $X_i$ and $r'$ cliques of type $Y_i$ and merge them pairwise. Since $(X_i,Y_i)$ is feasible, each merge produces a clique with label multiset $X_i\cup Y_i$; if $|X_i\cup Y_i|=d$ then each such merge contributes one clique of size $d$, and otherwise it contributes one incomplete clique of type $X_i\cup Y_i$. Hence, it follows that $A_i^{\omega}[s]=1$.

    For the converse implication, assume that $A_i^{\omega}[s] = 1$, and let $\mathcal{P}$ be a partition of the vertices of the shifted disjoint union of $G_{x_1}$ and $G_{x_2}$ into cliques of size at most $d$ that satisfies the three conditions in the definition of $A_i^{\omega}$. Then $\mathcal{P}$ results from a partition counted by $A_{i-1}^{\omega'}[s']$ by merging $r'$ pairs of cliques of type $X_i$ and $Y_i$ for some value $r'\in [\min\{s'(X_i), s'(Y_i)\}]_0$, where $\omega' = \omega$ if $|X_i\cup Y_i|\leq d-1$ and $\omega'=\omega-r'$ otherwise, and $s'$ is the state obtained from $s$ by subtracting $r'$ from $s(X_i\cup Y_i)$ and adding $r'$ to $s(X_i)$ and $s(Y_i)$.
    By the induction hypothesis for $i-1$, the algorithm sets $T_{i-1}^{\omega'}[s']=1$. Since the description of $s'$ and $\omega'$ matches the update rule of the algorithm, it follows that $T_i^{\omega}[s]=1$.

    We now define the auxiliary table $\hat{A}^{\omega}\in \{0,1\}^{\states^{2k}}$ as follows. Let $\hat{G}$ be the graph resulting from the disjoint union of $G_{x_1}$ and $G_{x_2}$ by adding all edges as defined by $\alpha$. Then $\hat{A}^{\omega}[s]=1$ if and only if there exists a partial solution in $\hat{G}$ of weight $\omega$ and state $s$. Clearly, $\hat{G}$ results from $G^h$ by applying the relabeling $\gamma^R$. Hence, a partition of $G^h$ contributes has state $s'$ and weight $\omega$ if and only if the same partition has state $s'_{\rightarrow \gamma^R}$ and weight $\omega$ in $\hat{G}$.
    
    Therefore, for every state $s\in\states$, it holds that
    \[
        \hat{A}^{\omega}[s] = \orstack{s'\in\states^{2k}\\ (s')_{\rightarrow \gamma^R}=s\downop} A_h^{\omega}[s'].
    \]
    This is exactly the same invariant as in the definition of the algorithmic table $\hat{T}_0^{\omega}$. Since we already proved that $A_h^{\omega}=T_h^{\omega}$ for every $\omega\in [n]_0$, it follows that $\hat{A}^{\omega}=\hat{T}_0^{\omega}$ for every $\omega\in [n]_0$.

    Finally, we show that $A_x^{\omega}=T_x^{\omega}$.
    Since $G_x$ results from the graph $\hat{G}$ by applying the relabeling $\beta$, it holds for each $s\in \states$ that
    \[
        A_x^{\omega}[s] = \orstack{s'\in\states\\ (s')_{\rightarrow \beta}=s} \hat{A}^{\omega}[s'].
    \]
    Since this is exactly the definition of $T_x^{\omega}$ from $\hat{T}_0^{\omega}$, and since we already proved that $\hat{A}_0^{\omega}=\hat{T}_0^{\omega}$, it follows that $A_x^{\omega}=T_x^{\omega}$.
\end{proof}

The following corollary follows then immediately from \cref{obs:ub-A-sol}.

\begin{corollary}\label{cor:ub-sol}
    The graph $G$ admits a $d$-clique packing of size at least $\target$ if and only if there exists a state $s\in \states$ and an integer $t'\geq \target$ such that $T_r^{t'}[s] = 1$.
\end{corollary}

Now we bound the running time of the algorithm.

\begin{lemma}\label{lem:ub-time}
    All tables $T_x^{\omega}$ can be computed in time $n^{\Oh(k^{d-1})}$.
\end{lemma}

\begin{proof}
    Let $N \coloneq |\states^{2k}|$.
    Since $\mathcal{M}^{2k}_d$ contains all multisets of size at most $d-1$ over $[2k]$, it holds that
    \[
    |\mathcal{M}^{2k}_d| \leq \sum_{j=0}^{d-1} (2k)^j = \frac{(2k)^d-1}{2k-1} 
    \leq \frac{2^d k^d}{2k - k} = \Oh(k^{d-1}).
    \]
    It follows that $N = (n+1)^{|\mathcal{M}^{2k}_d|} = n^{\Oh(k^{d-1})}$.
    Since $\states \subseteq \states^{2k}$, we also have $|\states| \leq N$.

    Now we show for every node $x\in\nodes$, that all tables $T_x^{\omega}$ for all values of $\omega$ can be computed in polynomial time in $N$ assuming that the tables $T_{x'}^{\omega'}$ for all children $x'$ of $x$ and all values $\omega'$ are given. Since we assume that the size of $\syntaxtree$ is bounded polynomially in $n$, this implies that the whole algorithm runs in polynomial time in $N$ following a bottom up dynamic programming scheme, and thus in time $n^{\Oh(k^{d-1})}$, which proves the claim.

    The bound is trivial for an introduce node $x$, since we can set the values $T_x^{\omega}$ directly by iterating over all values of $\omega$ and all entries $s\in \states$.

    Consider now a join node $x$. The number of values of $\omega$ is $n+1$, and the number of pairs $(X_i,Y_i)$ we iterate over is bounded by $h = N^2$.
    The table $T_0^{\omega}$ is computed by considering pairs of states $s_1,s_2\in \states$. Thus, for each fixed value of $\omega$, this requires at most $|\states|^2\leq N^2$ combinations, and therefore all tables $T_0^{\omega}$ are computed in time polynomial in $N$.

    Now some value $i\in [h]$, and assume that we have computed the tables $T_{i-1}^{\omega'}$ for all $\omega'\in [n]_0$. If $(X_i,Y_i)$ is not feasible, then $T_i=T_{i-1}$. Otherwise, for each value $\omega\in [n]_0$ and each state $s\in \states^{2k}$, the algorithm iterates over all $r'\in [\min\{s(X_i), s(Y_i)\}]_0$. Since $s(X)$ is bounded by $n$ for each state $s$ and a set $X$, there are at most $n+1$ choices for $r'$. For each such choice, the resulting state $s'$ is obtained by modifying a constant number of coordinates of $s$. Hence one transition from $T_{i-1}$ to $T_i$ requires polynomial time in $N$. Since there are $h$ such transitions, all auxiliary tables $T_i^{\omega}$ are computed in time polynomial in $N$.

    Finally, the table $T_x^{\omega}$ is obtained from $T_h^{\omega}$ by first computing $\hat{T}^{\omega}$ via the relabeling $\gamma^R$, then restricting it to $\hat{T}_0^{\omega}$ on $\states$, and then applying the relabeling $\beta$. Each of these three steps takes at most polynomial time in $N$, by iterating over all relevant states and updating the corresponding entry of the next table. Hence $T_x^{\omega}$ can be computed from $T_h^{\omega}$ in polynomial time in $N$. This concludes the proof.
\end{proof}

\begin{proof}[Proof of ~\cref{theo:ub}]
    Given the instance $(G,\target)$, together with a $k$-clique-expression of $G$ for some positive integer $k$, we first compute a $k$-NLC-expression of $G$. As discussed in the preliminaries, this can be done in polynomial time. Let $\syntaxtree$ be the syntax tree of the computed NLC-expression, and let $\nodes$ be the set of nodes of $\syntaxtree$.
    Our algorithm then computes all tables $T_x^{\omega}$ for each node $x\in \nodes$ and each value $\omega \in [n]_0$
    using \cref{ub-alg}. It outputs YES, if there exists some mapping $s\in \states$ with $T_r^{\target}[s]=1$, where $r$ is the root of $\syntaxtree$, and NO otherwise. The correctness of the algorithm follows from \cref{cor:ub-sol}. The algorithm runs in time $n^{\Oh(k^{d-1})}$ by \cref{lem:ub-time}.
\end{proof}

\section{Lower Bound}\label{section:lowerbound}

Let us fix some integer $d\geq 3$. In this section, we aim to prove that the $d$-\ClqPartp problem is $W[1]$-hard when parameterized by the linear clique-width, and that it cannot be solved in time $n^{o(k^{d-1})}$ assuming ETH, even when a linear $k$-clique-expression is given as part of the input.
We base our lower bound on a reduction from the \textsc{Multi-Colored Clique} problem defined as follows.

\problemdef{\textsc{Multi-Colored Clique}}
{A graph $G$ with a partition of its vertex set into $k$ color classes.}
{Is there a clique in $G$ containing one vertex from each partition?}

We also use a result from Lokshtanov et al.~\cite{DBLP:journals/eatcs/LokshtanovMS11} that states the following.

\begin{lemma}[Theorem 5.2.]\label{lem:multi-colored-clique}
    Assuming ETH, there is no $f(k)n^{o(k)}$ time algorithm for the \textsc{Multi-Colored Clique}, even on instances with the same number of vertices in each color class and the same number of edges between each pair of color classes, where $k$ is the number of colors and $n$ is the number of vertices in each color class.
\end{lemma}

However, we will restrict to a special case of this problem, that we call $d$-\textsc{Binomial Multi-Colored Clique} for some positive integer $d$, where the number of colors is a binomial coefficient $\binom{k}{d}$ for some integer $k\in \mathbb{N}$. We show that the lower bound holds for this restricted problem as well.

For each positive integer $d$, let $\mathcal{H}_d$ be the family of all $\binom{k}{d}$-partite graphs with $n$ vertices in each partition, and $m$ edges between each pair of partitions for some arbitrary integers $n,m,k\in \mathbb{N}$. 
We show that this problem stays hard even when restricted to instances in $\mathcal{H}_d$ for any positive integer $d$. Formally, we define the restricted version as follows:

\problemdef{$d$-\textsc{Binomial Multi-Colored Clique}}
{A graph $H\in \mathcal{H}_d$.}
{Is there a clique in $H$ containing one vertex in each partition?}

\begin{corollary}\label{cor:binomial-multi-colored-clique}
    For every positive integer $d$, the $d$-\textsc{Binomial Multi-Colored Clique} problem is $W[1]$-hard parameterized by $k$, and cannot be solved in time $n^{{o\big(\binom{k}{d}\big)}}$ assuming ETH, where $\binom{k}{d}$ is the number of partitions of the input graph.
\end{corollary}

\begin{proof}
    We show first that there exists a fixed integer $k_0$ (that depends on $d$ only), such that for each integer $k\geq k_0$ there exists an integer $k'$ with $k\leq \binom{k'}{d} < 2k$.
    We choose $k_0$ to be the smallest integer such that for the smallest integer $k'_0$ with $\binom{k'_0}{d}\geq k$ it holds that $k'_0 \geq 2d$.
    Then it holds for each $k\geq k_0$ and for $k'$ the smallest integer with $\binom{k'}{d}\geq k$ that $k'\geq k'_0\geq 2d$, and thus 
    \begin{align*}
        &\frac{\binom{k'}{d}}{k}\leq\frac{\binom{k'}{d}}{\binom{k'-1}{d-1}}\\
        &= \frac{k'!}{d!(k'-d)!}\frac{(d-1)!(k'-1-d)!}{(k'-1)!}\\
        &= \frac{k'}{k'-d} = \frac{k'-d+d}{k'-d} = 1 + \frac{d}{k'-d} \leq 2.
    \end{align*}

    Now we prove the claim. Assume that there exists an algorithm that solves the $d$-\textsc{Binomial Multi-Colored Clique} problem in time $n^{o\big(\binom{k'}{d}\big)}$ on graphs with $\binom{k'}{d}$ partitions, where each 
    color class has the same number of vertices and the same number of edges between each pair of color classes.

    Given an arbitrary instance of the \textsc{Multi-Colored Clique} problem with $k$ colors, such that each color class has $n$ vertices and there are $m$ edges between each pair of color classes, for some integers $n,m,k \in \mathbb{N}$, we construct the instance of the $d$-\textsc{Binomial Multi-Colored Clique} problem as follows:
    If $k = \binom{k'}{d}$ for some integer $k'$, then we use the same instance as the input. Otherwise, we create an instance with $\binom{k'}{d}$ color classes where $k'$ is the smallest integer with $\binom{k'}{d}\geq k$, 
    by adding $\binom{k'}{d}-k$ new partitions to the input graph, each containing $2n$ vertices. We also add $n$ vertices to each of the original $k$ partitions. We make each of the first $n$ vertices of the added partitions adjacent to the original $n$ vertices of each of the original $k$ partitions, and to the first $n$ vertices of each of the other added partitions. Finally, for any pair of original partitions, we add arbitrary edges between the vertices added to these partitions to increase the number of edges between them to $n^2$. This results in an instance with $\binom{k'}{d}$ partitions, where each partition has $2n$ vertices, and there are $n^2$ edges between each pair of partitions.
    
    We claim that the resulting instance admits a clique with a vertex from each partition if and only if the original instance admits one. Clearly, such a clique must contain a vertex from the first $n$ vertices of each added partition, since the other $n$ vertices are isolated. Since there exist at least one added partition, and since each such vertex is adjacent to all and only to the original $n$ vertices of each original partition, such a clique must contain a vertex from each of the original $k$ partitions, and thus it is a multi-colored clique in the original instance. Conversely, each such clique can be extended to a clique in the final graph by including the first vertex of each added partition.

    Hence, by applying the assumed algorithm for the $d$-\textsc{Binomial Multi-Colored Clique} problem on the constructed instance, we can solve the original instance of the \textsc{Multi-Colored Clique} problem in time $(2n)^{o\big(\binom{k'}{d}\big)}$ = $n^{o\big(k\big)}$ since we can assume that $\binom(k') {d} < 2k$. This contradicts \cref{lem:multi-colored-clique}, and thus the claim follows.

    Finally, the $W[1]$-hardness of the $d$-\textsc{Binomial Multi-Colored Clique} problem follows from the fact that the \textsc{Multi-Colored Clique} problem is $W[1]$-hard parameterized by $k$, and that the reduction described above is a parameterized reduction.
\end{proof}

In the rest of this section, we provide a reduction from the $(d-1)$-\textsc{Binomial Multi-Colored Clique} problem to the $d$-\ClqPartp problem on graphs of bounded clique-width that proves the claimed lower bounds.
Let $H$ be an instance of the $(d-1)$-\textsc{Binomial Multi-Colored Clique} problems with $r\coloneq \binom{k}{d-1}$ partitions. Let $n$ be the number of vertices in each partition, and $m$ be the number of edges between each pair of partitions.

We define the family of \emph{colors} $R\coloneq \binom{[k]}{d-1}$, and we call each subset $\{i_1, \ldots, i_{d-1}\}\in R$ a \emph{color}. We assign a distinct color to each partition of $H$ (arbitrarily), where we denote the partition with color $\{i_1, \ldots, i_{d-1}\}$ by $V_{\{i_1, \ldots, i_{d-1}\}}$. We also call each partition a \emph{color class}. For each color $C\in R$, let $V_C = \{u_C^1, \ldots, u_C^n\}$ be an arbitrary but fixed enumeration of the vertices in the color class $V_C$.

Our goal is to build an instance $\hat{G}$ of the $d$-\ClqPartp problem with bounded clique-width $\Oh(k)$ such that $H$ has a multi-colored clique if and only if the constructed instance is a YES-Instance. In order to achieve this, we first define an auxiliary graph $G$ with the latter property. Then we define the graph $\hat{G}$ of bounded clique-width from $G$ by adding edges to it, and we show that $G$ has a $d$-clique partition if and only if $\hat{G}$ has one.

The graph $G$ consists of a collection of the so called gadgets---graphs of constant clique-width that serve different functionalities. We add different copies of each gadget representing different aspects of the input instance. Mainly, we have two kinds of gadgets, each consisting of ``simpler'' gadgets. First, we define \emph{color gadget}s that represent the different partitions of the multicolored clique instance.
We add $r$ color gadgets, each corresponding to a color in $R$, and it encodes the id of the vertex selected in the corresponding color class in a hypothetical multi-colored clique. 
A color gadget consists of a sequence of \emph{id gadget}s and \emph{cuts} between them. An id gadget allows the selection of a specific vertex in a partition by its id (in $[n]$), while cuts ensure that consecutive id gadgets select the same id.

Second, we define \emph{edge gadget}s, which ensure that the selected ids of different partitions are adjacent. An edge gadget consists of $m$ \emph{simple edge gadget}s that represent the single edges between two color classes, and a \emph{selection gadget} that ``consumes'' all but a single simple edge gadget that will be ``consumed'' by the id gadgets adjacent to it. Now we provide a formal description of these gadgets:

\subparagraph{Color gadget.}
The construction consists of $r$ \emph{color gadgets}, each corresponding to a color in $R$. A color gadget is a graph that encodes the choice of a vertex in the corresponding color class. It has a path-like structure, as it consists of a sequence of gadgets called \emph{id gadgets} connected by small independent sets. While an id gadget encodes the choice of a vertex in the corresponding color class, the independent sets between id gadgets ensure that the choice is copied correctly from one id gadget to the next. This gives the graph $G$ a grid-like structure, where the rows correspond to colors and each column corresponds to a pair of colors, and ensures that the chosen vertices in these two color classes are adjacent in $H$.

More concretely, let $\mathscr{L} \coloneq \binom{R}{2}$ be the set of pairs of colors, and $\ell \coloneq |\mathscr{L}|$.  We define a \emph{color gadget} as the graph consisting of a sequence of $2\ell$ copies of a gadget called a \emph{color column}, where a color column consists of $n$ $(d-1)$-cliques. We add an independent set of size $n$ between each two consecutive color columns, and make its vertices adjacent to all vertices of both columns on its sides.
We also add an independent set of size $n$ adjacent to all vertices of both the first and the last color columns.
Hence, each color gadget consists of $2\ell$ color columns and $2\ell$ independent sets in an alternating manner. We partition a color gadget into $\ell$ segments, each consisting of two consecutive color columns together with the independent set between them. We call each such segment an \emph{id gadget}. We call each independent set inside an id gadget an \emph{inner cut}, and each independent set between two id gadgets an \emph{outer cut}.

Note that each element of $\mathscr{L}$ is an unordered pair of subsets of $[k]$ of size $d-1$ each. We fix an arbitrary enumeration of the elements $\mathscr{L}$ as $\mathscr{L}=\{L_1, \ldots, L_\ell\}$. We denote the $q$-th id gadget in the color gadget $\colgad_{\{i_1,\dots i_d\}}$ by $\idgad_{\{i_1,\dots,i_d\}}^q$. Hence, each color gadget forms a row in the grid, and for each integer $q\in [\ell]$, the family consisting of the $q$-th id gadget of each color gadget forms the $q$-th column in the grid. On a more fine-grained view, each row of the grid consists of $n$ finer rows, we call them \emph{unit rows}, where each unit row consists of a sequence of $2\ell$ $(d-1)$-cliques, and $n$ single vertices between them in an alternating manner, where each such single vertex belongs to an independent set of the color gadget.

Intuitively, the choice of whether such a vertex is attached to the clique preceding it or the clique succeeding it in the unit row encodes a Boolean value, where the accumulation of all these values over an id gadget encodes the id of the vertex chosen in the corresponding color class.

\subparagraph{Edge gadget.}
In order to ensure that the chosen vertices in the color classes are adjacent, for each pair of colors $L_q=\{C_1, C_2\}$, we modify the $q$-th id gadget of the color gadgets $\colgad_{C_1}$ and $\colgad_{C_2}$ as follows. First we remove the inner cuts from both gadgets. Then we build the so called \emph{edge gadgets} between these two id gadgets as follows: we start by defining a \emph{simple edge gadget} that corresponds to a single edge $\{u_{C_1}^{j^1}, u_{C_2}^{j^2}\}$ in $H$. A simple edge gadget consists of an independent set of size $2n$ with vertices $v_1,\dots v_n$ and $w_1, \dots w_n$. The vertices $v_1,\dots v_{j_1}$ are adjacent to all vertices of the first color column of $\colgad_{C_1}^q$, and the vertices $v_{j_1+1},\dots v_n$ are adjacent to all vertices of the second color column of $\colgad_{C_1}^q$. Similarly, the vertices $w_1,\dots w_{j_2}$ are adjacent to all vertices of the first color column of $\colgad_{C_2}^q$, and the vertices $w_{j_2+1},\dots w_n$ are adjacent to all vertices of the second color column of $\colgad_{C_2}^q$. 

As mentioned before, an edge gadget corresponding to the pair of colors $L_q=\{C_1, C_2\}$ consists of the $m$ simple edge gadgets $\simpegad_1,\dots\simpegad_m$ corresponding to the $m$ edges between the two color classes, together with a $\mathcal{F}$-selection gadget added to them, where $\mathcal{F}$ is the family $\big\{V(\simpegad_1), \dots, V(\simpegad_m)\big\}$.
A selection gadget allows to select a single simple edge gadget from the family by ``consuming'' all but one of them. Now we provide a formal description of the selection gadget.

\subparagraph{Selection gadget.}
Given a family $\mathcal{F}$ of $m$ sets of vertices $X_1,\dots X_m$, of size $r$ each where $r$ is even, such that the union of all these sets forms an independent set, an $\mathcal{F}$-selection gadget $\selgad(\mathcal{F})$ consists of $m-1$ copies of a structure called $r$-\emph{selection column} $C_1$,\dots $C_{m-1}$, where $C_i$ is added between $X_i$ and $X_{i+1}$.

The selection column $C_i$ is a simple cycle on $2r$ vertices. We partition the vertices of this cycle into $4$ sets by assigning the colors $c_0, \dots c_3$ to them, where the color assigned to each vertex is defined by its id modulo $4$. Let $v_1,\dots v_r$ be the vertices of $X_i$ and $w_1,\dots w_r$ be the vertices of $X_{i+1}$ we add all edges between the vertices $v_i$ for odd values of $i$ and the vertices colored $c_1$ and $c_2$, and all edges between the vertices $v_i$ for even values of $i$ and the vertices colored $c_3$ and $c_0$. Similarly, we add all edges between the vertices $w_i$ for odd values of $i$ and the vertices colored $c_2$ and $c_3$, and all edge between the vertices $w_i$ for even values of $i$ and the vertices colored $c_1$ and $c_0$.

The even cycle structure of the selection column allows to decompose it into two perfect matchings, a ``horizontal'' one, containing the edges between the vertices colored $c_0$ and $c_1$ and the edges between the vertices colored $c_2$ and $c_3$, and a ``vertical'' one, containing the rest of the edges. Intuitively, we connect the endpoints of the horizontal edges to vertices of $X_i$, and the endpoints of the vertical edges to vertices of $X_{i+1}$. We separate the vertices of $X_i$ and $X_{i+1}$ into even and odd indices to avoid conflicts, where we connect the even indexed vertices to the edges of even indices of either matching, and the odd indexed vertices to the edges of odd indices.

Finally, we add $r$ cliques of size $d-3$ each, and make each of these cliques adjacent to all vertices of the selection column and all vertices of the sets in $\mathcal{F}$. This ensures that any triangle formed between $\mathcal{F}$ and the selection column can be extended to a clique of size $d$ and vice versa.

This gadget is inspired by a gadget that was recently introduced by Bojikian et al.~\cite{DBLP:conf/esa/BojikianCK25} for triangle packing parameterized by cutwidth.
Intuitively, any solution $\mathcal{P}$ must pack the edges of exactly one of the two perfect matchings of a selection column in its cliques, as once one edge of either matching is packed in some $d$-clique, all the other edges must agree. Hence, $\mathcal{P}$ either includes the edges of the horizontal matching together with the vertices of $X_i$, or the edges of the vertical matching together with the vertices of $X_{i+1}$. The endpoints of the matching edges then, together with a vertices of $X_i$ or $X_{i+1}$, form triangles that can be extended to cliques of size $d$ by adding the $(d-3)$-cliques adjacent to the selection columns. Since we add $m-1$ selection columns, $\mathcal{P}$ must pack all but one simple edge gadgets with the selection columns, keeping exactly one simple edge gadget free to be packed with the id gadgets adjacent to it.
We note that while connecting each edge of a selection column to a single vertex of $X_i$ or $X_{i+1}$ results in a ``cleaner'' construction, these ``redundant'' edges are required to ensure that the resulting graph has bounded clique-width, as they allow to divide the vertices of the selection column into four equivalent classes as described above.

This concludes the description of a selection gadget, and thus the description of the graph $G$. Now we prove the correctness of our reduction.

\begin{figure}
    \centering
    \begin{subfigure}{.3\textwidth}
        \centering
        \includegraphics[height=3.2cm]{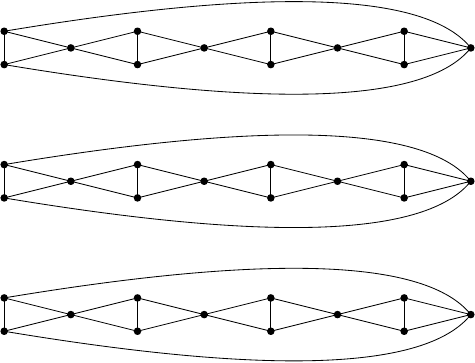}\\
        $(a)$
    \end{subfigure}
    \hfill     \begin{subfigure}{.3\textwidth}
        \centering
        \includegraphics[height=3.2cm]{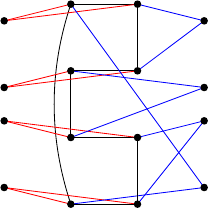}\\
        $(b)$
    \end{subfigure}
    \hfill     \begin{subfigure}{.3\textwidth}
        \centering
        \includegraphics[height=3.2cm]{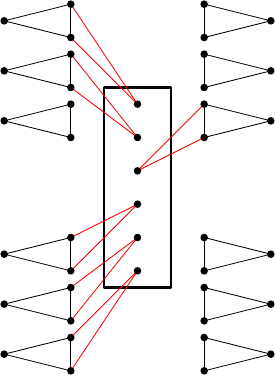}\\
        $(c)$
    \end{subfigure}
    \caption{
        In $(a)$ we depict three ``unit rows'' of a \emph{color gadget} for $d=3$. Each unit row consists of an alternation of $2$-cliques and single vertices.
        In $(b)$ we depict a part of a \emph{selection gadget}.
        We depict two simple edge gadgets (the independent sets on the left and right) and a selection column between them. We omit the redundant edges for clarity, connecting each vertex of a simple edge gadget to the endpoints of a single edge. Note that the vertices of the first simple edge gadget build triangles with the edges of the horizontal matching of the selection column, while the vertices of the second simple edge gadget build triangles with the edges of the vertical matching.
        Finally, in $(c)$ we depict a simple edge gadget replacing the inner cuts of two id gadgets. In this example, the simple edge gadget corresponds to the edge $\{u_{C_1}^2, u_{C_2}^3\}$ in $H$.}
\end{figure}

\begin{lemma}\label{lem:lb-prt-if-clq}
    If $H$ has a multi-colored clique, then $G$ has a $d$-clique partitioning.
\end{lemma}

\begin{proof}
    Let $Q$ be a multi-colored clique in $H$. For each color $C\in R$, let $u_C^{i_C}\in V_C$ be the unique vertex of $Q$ in the color class $V_C$, where $i_C\in [n]$ is its id. We show how to partition $V(G)$ into cliques of size exactly $d$.

    We start with the color gadgets. Fix a color $C\in R$, and consider an id gadget $\idgad_C^q$ whose inner cut is present. Let $L_C^q$ and $R_C^q$ be its left and right color columns. Each of these columns consists of $n$ vertex-disjoint $(d-1)$-cliques; we denote them by $L_C^q(1),\dots,L_C^q(n)$ and $R_C^q(1),\dots,R_C^q(n)$. Let $I_C^q=\{x_1,\dots,x_n\}$ be the inner cut of $\idgad_C^q$. For each $t\in [n]$, we pack $x_t$ together with
    $L_C^q(t)$ in a $d$-clique if it holds that $t\leq i_C$, and with $R_C^q(t)$ otherwise.
    Since every vertex of the inner cut is adjacent to all vertices of both neighboring columns, this yields $n$ cliques of size $d$. Hence exactly the first $i_C$ vertices of the inner cut are packed with the left column, and the remaining $n-i_C$ vertices are packed with the right column.

    Now let $O$ be an outer cut of the same color gadget, lying between two consecutive id gadgets. We pack $O$ in the reverse way: for $O=\{x_1,\dots,x_n\}$ and its neighboring columns are $L$ on the left and $R$ on the right (for the outer cut added the end of each color gadget, we consider the last color column as its left, and the very first one as its right), then for each $t\in [n]$ we put $y_t$ together with
    $L(t)$ if $t> i_C$, and with $R(t)$ otherwise.
    Thus an outer cut sends the first $i_C$ vertices to the right and the remaining vertices to the left. Applying this rule to every outer cut shows that the same value $i_C$ is propagated through the whole color gadget. In particular, every $(d-1)$-clique in a color column that does not belong to an edge gadget is extended by exactly one cut vertex to a $d$-clique.

    Next consider a pair of colors $L_q=\{C_1,C_2\}\in \mathscr{L}$ and the corresponding edge gadget. Since $Q$ is a clique in $H$, the selected vertices $u_{C_1}^{i_{C_1}}$ and $u_{C_2}^{i_{C_2}}$ are adjacent in $H$. Let $\simpegad_q$ be the simple edge gadget corresponding to that edge, and 
    Let $v_1,\dots,v_n$,$w_1,\dots,w_n$ be the vertices of $\simpegad_q$,    
    where the vertices $v_t$ are adjacent to the id gadget of $C_1$ and the vertices $w_t$ are adjacent to the id gadget of $C_2$. We use this simple edge gadget instead of the deleted inner cuts. For each $t\in [n]$, we pack $v_t$ together with the $t$-th $(d-1)$-clique of the first column of $\idgad_{C_1}^q$ if $t\leq i_{C_1}$, and with the $t$-th $(d-1)$-clique of the second column otherwise. Analogously, we pack $w_t$ together with the $t$-th $(d-1)$-clique of the first column of $\idgad_{C_2}^q$ if $t\leq i_{C_2}$, and with the $t$-th $(d-1)$-clique of the second column otherwise. These adjacencies exist by the definition of a simple edge gadget. Therefore the chosen simple edge gadget realizes exactly the same left-right split as the deleted inner cuts of both id gadgets.

    It remains to cover the vertices inside the selection gadgets. Consider the selection gadget attached to the family $\mathcal{F}=\{X_1,\dots,X_m\}$ of simple edge gadgets for the pair $L_q$, and let $X_p=V(\simpegad_q)$ be the set corresponding to the chosen edge. We first construct a triangle packing between $\mathcal{F}\setminus\{X_p\}$ and the selection columns. For each $j<p$, we use the selection column between $X_j$ and $X_{j+1}$ to cover $X_j$, and for each $j\geq p$, we use the selection column between $X_j$ and $X_{j+1}$ to cover $X_{j+1}$. 
    
    Fix one such column $C_j$. In the former case ($j<p$), let $X_j = v_1,\dots,v_r$. By construction of the selection column, every odd-indexed vertex $v_t$ is adjacent to all vertices of the color classes $c_1$ and $c_2$, while every even-indexed vertex $v_t$ is adjacent to all vertices of the color classes $c_3$ and $c_0$. Hence, for each $t\in[r]$, we can build a triangle between $v_t$ and the endpoints of the $(2t-1)$st edge of the cycle (as the endpoints of such edges alternate in colors between $1,2$ and $3,0$). Similarly, for the latter case ($j\geq p$), let $X_{j+1} = v_1, \dots, v_r$.
    Then every odd-indexed vertex $v_t$ is adjacent to all vertices of the color classes $c_2$ and $c_3$, while every even-indexed vertex $v_t$ is adjacent to all vertices of the color classes $c_1$ and $c_0$. Hence, for each $t\in[r]$, we can build a triangle between $v_t$ and the endpoints of the $(2t)$th edge of the cycle (as the endpoints of such edges alternate in colors between $2,3$ and $0,1$). In both cases, we obtain $r$ vertex-disjoint triangles that cover exactly the vertices of the chosen set together with all vertices of $C_j$. Doing this for every selection column yields a triangle packing that covers all vertices of the selection gadget and all vertices in $\mathcal{F}\setminus\{X_p\}$, while leaving the vertices of $X_p$ uncovered inside the selection gadget.

    Finally, for each triangle in this packing, we choose one of the attached cliques of size $d-3$ and add all its vertices to the triangle. Since every such clique is adjacent to all vertices of the corresponding selection column and all vertices of the sets in $\mathcal{F}$, each triangle extends to a clique of size $d$. As the attached $(d-3)$-cliques are pairwise disjoint and their number is exactly $r$ per selection column, this turns the triangle packing above into a $d$-clique partition of the whole selection gadget together with all sets $X_j$ for $j\neq p$.

    We perform the above construction independently for every color gadget, and for every edge gadget corresponding to every pair of colors. It is not hard to see that both inner cuts and simple edge gadgets consume the same number of $(d-1)$-cliques from the left color column of each id gadget, and the same number of $(d-1)$-cliques from the right color column, where the outer cuts complement these selections by consuming the remaining $(d-1)$-cliques. It also holds that  every non-chosen simple edge gadget is absorbed by its selection gadget. Consequently, every vertex of $G$ belongs to exactly one clique of size $d$. Hence $G$ has a $d$-clique partitioning.
\end{proof}

Now we aim to to prove the converse implication.
We first prove some properties that every $d$-clique partitioning of $G$ has. These properties will be useful to show that the selected vertices correspond to a multi-colored clique in $H$.
From now on, we assume that $G$ admits a $d$-clique partitioning, and let $\mathcal{P}$ be such a partition. Let $E(\mathcal{P})$ be the set of all edges in the cliques of $\mathcal{P}$. We start with the following definition:

\begin{definition}
    Given an id gadget $\idgad$ in $G$, we define the \emph{state} of $\idgad$ defined by $\mathcal{P}$, denoted by $s(\idgad)$, as the number of $(d-1)$-cliques in the left color column of $\idgad$ that are contained in cliques of $\mathcal{P}$ together with vertices of the inner cut of this gadget if such an inner cut exists, or with vertices of the edge gadget replacing the inner cut otherwise.
\end{definition}

Now we prove that a selection gadget must pack the vertices of all but exactly one of the simple edge gadgets internally, keeping a single simple edge gadget ``exposed'' to be packed together with id gadgets.

\begin{lemma}\label{lem:lb-selgad}
    Let $\selgad(\mathcal{F})$ be a selection gadget in $G$ with $\mathcal{F} = \{X_1,\dots X_m\}$, where all sets of $\mathcal{F}$ have size $r$. Then there exists a unique value $i\in [m]$ such that for all $j\neq i$, all vertices of $X_j$ belong to cliques of $\mathcal{P}$ that lie completely in the selection gadget, while the vertices of $X_i$ are in cliques that only intersect the selection gadget in $X_i$ itself.
\end{lemma}

\begin{proof}
    Fix an arbitrary selection column $C_j$ for $j\in[m-1]$.
    First note that no clique of $\mathcal{P}$ can contain a single vertex of $C_j$, since each vertex of $C_j$ is only adjacent to the $(d-3)$-cliques, to two other vertices of $C_j$, and to the vertices of $X_j$ and $X_{j+1}$ that are pairwise non-adjacent---as they induce an independent set by definition.
    On the other hand, since $C_j$ induces a simple cycle in $G$, each clique of $\mathcal{P}$ contains at most two vertices of $C_j$ that appear consecutively on the cycle. Therefore, the intersection of $E(\mathcal{P})$ with $C_j$ must be one of the two perfect matchings of $C_j$, either the ``horizontal'' one or the ``vertical'' one.

    By construction of the gadget, vertices of $X_j$ are only adjacent to endpoints of the edges of the horizontal matching, while vertices of $X_{j+1}$ are only adjacent to endpoints of the edges of the vertical matching. 
    Therefore:
    If $\mathcal{P}$ packs the edges of the horizontal matching, then all cliques covering $C_j$ use vertices of $X_j$ (and none of $X_{j+1}$),
    and if $\mathcal{P}$  packs the edges of the vertical matching, then all cliques covering $C_j$ use vertices of $X_{j+1}$ (and none of $X_j$). Hence, exactly $r$ vertices are taken from exactly one of the two adjacent sets consuming all vertices of $C_j$ and all vertices of the selected set, while no vertex of the other set is used.

    Therefore, there must exist a unique index $i\in [m]$ such that for all $j<i$, the columns $C_j$ cover all vertices of $X_j$ and for all $i\leq j \leq m-1$, the columns $C_j$ cover all vertices of $X_{j+1}$. Hence, all vertices of $X_i$ are not covered by any column, and thus must be covered by cliques that only intersect the selection gadget in $X_i$ itself.
\end{proof}

The following lemma states that all id gadgets of the same color gadget have the same state. Intuitively, this holds since the $(d-1)$-cliques of a color column that are not consumed by the inner cut vertices, must be consumed by the outer cut vertices.

\begin{lemma}\label{lem:lb-same-state}
    Let $\idgad, \idgad'$ be id gadgets in the same color gadget. Then it holds that $s(\idgad) = s(\idgad')$.
\end{lemma}

\begin{proof}
    We prove the claim for two consecutive id gadgets of the same color gadget. The proof of the lemma then follows follows by induction.
    Let $\idgad^q$ and $\idgad^{q+1}$ be consecutive id gadgets in one fixed color gadget, and let $O^q$ be the outer cut between them. Let $s \coloneq s(\idgad^q)$.
    By definition of the state, exactly $s$ $(d-1)$-cliques of the left color column of $\idgad^q$ are completed together with the inner cut or with the replacing edge gadget. Hence exactly $n-s$ $(d-1)$-cliques of the right color column of $\idgad^q$ are completed with that inner cut or edge gadget respectively, and therefore the remaining $s$ $(d-1)$-cliques of that right column must be completed with vertices of the outer cut $O^q$.

    The set $O^q$ is an independent set of size $n$, and each vertex of $O^q$ is adjacent to all vertices of the right color columns of $\idgad^q$ and the left color columns of $\idgad^{q+1}$. Thus every clique of $\mathcal{P}$ containing a vertex of $O^q$ contains exactly one vertex of $O^q$ and one $(d-1)$-clique from one of its two neighboring columns. Therefore, exactly $n-s$ vertices of $O^q$ are with $(d-1)$-cliques of the left column of $\idgad^{q+1}$.

    Consequently, exactly $n-s$ $(d-1)$-cliques of the left color column of $\idgad^{q+1}$ are consumed by $O^q$. The remaining $s$ $(d-1)$-cliques of that left column must be completed with the inner cut of $\idgad^{q+1}$ or with its replacing edge gadget. By definition of the state of an id gadget, it follows that $s(\idgad^{q+1}) = s$, which proves the claim.
\end{proof}

Finally, we show that the edge gadgets ensure that the selected vertices in different color gadgets are adjacent in $H$. Again, this holds since a single simple edge gadget must be packed together with the id gadgets, and these simple edge gadgets ensure that the selected vertices in the two color gadgets are adjacent in $H$.

\begin{lemma}\label{lem:lb-edge-gadget}
    Let $\idgad_{C_1}^q$ and $\idgad_{C_2}^q$ be two id gadgets with the edge gadget corresponding to the pair of colors $L_q=\{C_1,C_2\}$ added between them. Let $s_1 = s(\idgad_{C_1}^q)$ and $s_2 = s(\idgad_{C_2}^q)$. Then the edge $\{u_{C_1}^{s_1}, u_{C_2}^{s_2}\}$ exists in $H$.
\end{lemma}

\begin{proof}
    Let $i_1 \coloneq s(\idgad_{C_1}^q)$ and $i_2 \coloneq s(\idgad_{C_2}^q)$. We show that $\{u_{C_1}^{i_1},u_{C_2}^{i_2}\}\in E(H)$.

    Let $\mathcal{F}=\{X_1,\dots,X_m\}$ be the family of vertex sets of the simple edge gadgets contained in the edge gadget between $\idgad_{C_1}^q$ and $\idgad_{C_2}^q$. By the previous lemma on selection gadgets, there exists a unique index $p\in [m]$ such that all vertices of $X_j$ for $j\neq p$ are covered by cliques lying completely inside the selection gadget, while the vertices of $X_p$ belong to cliques that intersect the selection gadget only in $X_p$. Hence the only simple edge gadget that can interact with the two id gadgets is the gadget corresponding to $X_p$.

    Let $X_p=V(\simpegad_p)$, and let $\{u_{C_1}^{j_1},u_{C_2}^{j_2}\}$ be the edge of $H$ represented by $\simpegad_p$. By construction of the simple edge gadget, its vertices are partitioned into vertices $v_1,\dots,v_n$ adjacent to $\idgad_{C_1}^q$ and vertices $w_1,\dots,w_n$ adjacent to $\idgad_{C_2}^q$, where $v_1,\dots,v_{j_1}$ are adjacent to all vertices of the first color column of $\idgad_{C_1}^q$ and $v_{j_1+1},\dots,v_n$ are adjacent to all vertices of the second color column, and analogously $w_1,\dots,w_{j_2}$ are adjacent to the first color column of $\idgad_{C_2}^q$ while $w_{j_2+1},\dots,w_n$ are adjacent to its second color column.

    Since the vertices of $X_p$ are not covered inside the selection gadget, every vertex $v_t$ must be packed with one of the two color columns of $\idgad_{C_1}^q$, and every vertex $w_t$ must be packed with one of the two color columns of $\idgad_{C_2}^q$. Therefore exactly $j_1$ $(d-1)$-cliques of the left column of $\idgad_{C_1}^q$ are completed using vertices of $X_p$, and the remaining $n-j_1$ such cliques are completed on the right side. By the definition of the state of an id gadget, this means that $s(\idgad_{C_1}^q)=j_1$. The same argument on the $C_2$-side yields $s(\idgad_{C_2}^q)=j_2$.

    Consequently, $\Big\{u_{C_1}^{s(\idgad_{C_1}^q)},u_{C_2}^{s(\idgad_{C_2}^q)}\Big\}=\big\{u_{C_1}^{j_1},u_{C_2}^{j_2}\big\}$ is exactly the edge of $H$ corresponding to the simple edge gadget $\simpegad_p$, and hence, the two states correspond to an edge of $H$.
\end{proof}

Combining all the properties above, we get the following:

\begin{lemma}\label{lem:lb-clq-if-prt}
    If $G$ has a $d$-clique partitioning $\mathcal{P}$, then $H$ has a multi-colored clique.
\end{lemma}

\begin{proof}
    We construct a multi-colored clique in $H$ from $\mathcal{P}$.
    For each color $C\in R$, consider the color gadget $\colgad_C$. By \cref{lem:lb-same-state}, all id gadgets contained in $\colgad_C$ have the same state. We denote this common value by $s_C$. We now select in the color class $V_C$ the vertex $u_C^{s_C}$.
    Let $S \coloneq \{u_C^{s_C} \mid C\in R\}$. We claim that that $S$ is a clique in $H$.
    
    Let $C_1,C_2\in R$ be two distinct colors, and consider the index $q\in[\ell]$ such that $L_q=\{C_1,C_2\}$. Since the two id gadgets $\idgad_{C_1}^q$ and $\idgad_{C_2}^q$ linked by the corresponding edge gadget, it follows then from \cref{lem:lb-edge-gadget} that $\Big\{u_{C_1}^{s(\idgad_{C_1}^q)}, u_{C_2}^{s(\idgad_{C_2}^q)}\Big\}\in E(H)$. Since $C_1$ and $C_2$ were chosen arbitrarily, it follows that $S$ is a clique in $H$.
\end{proof}

\subparagraph{Bounded clique-width graph.}

Now we describe the graph $\hat{G}$. For each id gadget $\idgad_C^q$ in $G$, and each $(d-1)$-clique of this gadget, let us fix a bijective assignment from $C$ to the vertices of this clique. We call the value assigned to a vertex by this assignment its \emph{color}. 
Let us denote the left color column of an id gadget $\idgad_C^q$ by $Q_C^q(L)$ and its right color column by $Q_C^q(R)$. For $X\in \{L, R\}$ and $q\in [\ell]$, we call the collection for all $(d-1)$-cliques of all color columns $Q_C^q(X)$ for all colors $C\in R$ the $(q, X)$-\emph{grid column}. For each vertex of an inner cut, out cut or simple edge gadget attached to a color column $Q_C^q(X)$, we add all edges between this vertex and all vertices of the $(q,X)$-grid column whose colors belong to $C$. We denote the resulting graph by $\hat{G}$.

\begin{lemma}
    The graph $\hat{G}$ has a $d$-clique partitioning if and only if $G$ has one.
\end{lemma}

\begin{proof}
    Since $\hat{G}$ is a supergraph of $G$ with the same set of vertices, if $G$ has a $d$-clique partitioning, then $\hat{G}$ also has one.
    For the other direction, we show that $G$ and $\hat{G}$ have the same set of $d$-cliques. Let $C$ be a $d$-clique in $\hat{G}$. If $C$ does not intersect a color column, then $C$ is clearly a $d$-clique in $G$ as well. 
    
    Otherwise, we claim that $C$ consists of $d-1$ vertices of color columns and a single vertex that does not belong to a color column. This holds since the vertices of color columns are partitioned into $(d-1)$-cliques without edges between them, and any two vertices that do not belong to color columns and adjacent to color column vertices are not adjacent themselves. So the only case a $d$-clique of $\hat{G}$ doesn't belong to $G$ is if it contains a vertex $w$ of an inner cut, out cut or simple edge gadget attached to a color column $Q_C^q(X)$ together with $(d-1)$-clique of a different color column in the same grid column $\idgad_{C'}^q$ with $C\neq C'$. However, since $|C|=|C'|=d$, and $C\neq C'$, $C'$ must contain a color that does not belong to $C$, and hence, the vertex $w$ is not adjacent to all vertices of that $(d-1)$-clique, contradicting the assumption that $C$ is a clique. Therefore, every $d$-clique of $\hat{G}$ is a $d$-clique of $G$, and thus if $\hat{G}$ has a $d$-clique partitioning, then $G$ also has one.
\end{proof}

\begin{lemma}\label{lem:lb-linear-cw}
    The graph $\hat{G}$ has linear clique-width at most $\Oh(k)$. Moreover, given the instance $H$, the graph $\hat{G}$ together with a linear $\Oh(k)$-clique-expression of $\hat{G}$ can be constructed in polynomial time.
\end{lemma}

\begin{proof}
    Since the syntax tree of a linear clique-expression has a caterpillar structure, where all inner nodes are union nodes with a single introduce child, we consider a linear clique expression as a sequence of operations, each adding a vertex of a specific label to the graph, relabeling some label, or adding edges between labels. We will consider each operation of this sequence a ``timestamp'', and denote the graph constructed after the first $t$ operations by $\hat{G}_t$.

    We will preserve labels that will be used ot specific purposes. We first preserve the first $2k$ labels, where the first $k$ labels $\ell_1, \dots, \ell_k$ are used to construct the left color columns of the id gadgets, and the following $k$ labels $r_1, \dots, r_k$. are used to construct the right color columns of the id gadgets. 
    We preserve the following $k$ labels $\ell^0_1, \dots, \ell^0_k$ for the very first color column of each color gadget. We also preserve $d-1$ labels $c_1, \dots c_{d-1}$ for creating the cliques of these color columns.

    We preserve $8$ labels $z_0, \dots, z_3, z'_0, \dots, z'_3$ and three extra labels $y_0$, $y$ and $y'$ for creating the selection columns of the selection gadgets, and the labels $e$ and $e'$ for the $(d-3)$-cliques attached to them. 
    We preserve the labels $s$ and $s'$ for creating independent sets and cliques. 
    We also preserve a ``forget'' label $f$. Vertices of label $f$ at a timestamp $t$ already have all their neighbors in $\hat{G}$ adjacent to them in $\hat{G}_t$.
    
    When we say we add an independent set or a clique of size $r$, we mean that we add $r$ vertices of label $s$, where in the case of a clique, we add each vertex with label $s'$, we join it with label $s$ and then relabel it to the label $s$. When we use this operation, we assume that no other vertex has label $s$ at that timestamp. After applying this operation, we will join the created independent set or clique with some other labels of $G_t$, and then relabel all its vertices, freeing both labels $t$ and $t'$. Similarly, when we say we create a $(d-1)$-clique of an id gadget $\idgad_{C}^q$, we mean that we add $d-1$ vertices of labels $c_1,\dots c_{d-1}$, we add all edges between all these labels, and then injectively relabel these labels to the labels $\{\ell^0_c\colon c\in C\}$ if this is the first color column of the color gadget of color $C$, to $\{\ell_c\colon c\in C\}$ if this is a left color column of an id gadget of color $C$, or to $\{r_c\colon c\in C\}$ if this is a right color column of an id gadget of color $C$.

    Now we describe the construction of $G$ step by step. The construction will be divided into phases, where in each phase we create a grid column. We preserve the following two invariants:
    
    \begin{invariant}
        At the end of the $q$th phase for $q\in[\ell]$, the graph $\hat{G}_t$ consists of the first $q$ id gadgets of each color gadget, together with the out cuts and edge gadgets between them. Each vertex of the first grid column with color $c$ has label $\ell^0_c$, and each vertices of the last grid column with color $c$ has label $r_c$. All other vertices have label $f$.
    \end{invariant}

    \begin{invariant}
        At each timestamp, there exists at most one grid column with vertices of labels $\{\ell_c\colon c\in C\}$ and at most one grid column with vertices of labels $\{r_c\colon c\in C\}$ for each color $C$.
    \end{invariant}

    The latter invariant ensures that we can add edges between the grid columns and the inner cuts, outer cuts and simple edge gadgets attached to them by a single join operation per gadget, without affecting the other grid columns, where the other grid columns either are already forgotten, or not constructed yet.

    Now we describe each phase. The phase starts by constructing the left color columns of the $q$-th id gadgets of each color gadget, by creating $(d-1)$-cliques one by one. After that, if this is not the first phase, we create the outer cut between the $q$-th and the $(q-1)$-th id gadgets of each color gadget corresponding to a color $C$, by adding an independent set of size $n$ and joining it with all labels $\ell_c$ and $r_c$ for $c \in C$. After that we forget all label $r_c$ for all $c\in[k]$, and then we create the right color columns of the $q$-th id gadgets of each color gadget, by creating $(d-1)$-cliques one by one.
    
    Finally, we create the edge gadget corresponding to $L_q$, as follows: we will create the simple edge gadgets one by one. For each we will ensure that the selection columns preceding and following it (if exist) are already created, where a vertex of the selection column preceding with color $i$ it will be assigned the label $z_i$, and a vertex of the selection column following it with color $i$ will be assigned the label $z'_i$. Remember that the color assigned to a vertex of a selection column is defined by its id on the cycle modulo $4$.
    
    The process begins by relabeling $z'_i$ to $z_i$ for all $i\in[4]$ and $e'$ to $e$, since if this was not the first selection gadget, then the selection column preceding the current simple edge gadget is the selection column following the previous simple edge gadget. After that we create the selection column following the current simple edge gadget, if exists, by creating the $(d-3)$-cliques and a cycle on $2r$ vertices and assigning labels $z'_0,\dots z'_3$ as follows:
    in each step we add a vertex of label $y'$, make it adjacent to label $y$, relabel $y$ to $z'_i$ (except for the very first vertex, we relabel it to $y_0$) for the appropriate $i$, and then we relabel $y'$ to $y$. Finally, for the last vertex, we join it to $y_0$ and relabel both $y_0$ and $y$ to $z'_1$ and $z'_0$ respectively. After that we create the $r$ $(d-1)$-cliques using the labels $s'$ and $s$, relabeling each to $e'$ after creating. After that we make these cliques adjacent to all vertices of the selection column by joining $e$ with $z'_i$ for all $i\in[4]$.

    Then we create the simple edge gadget consisting of four independent sets one by one, the definition of a simple id gadget, where for each we create an independent set of the appropriate size, we make it adjacent to all corresponding colors of the corresponding grid columns, selection columns preceding and following it and the $(d-3)$-cliques attached to them, and then we forget its labels (the vertices of this independent set must have label $s$ at this point). After we create a simple edge gadget we forget the labels of the preceding selection column and its $(d-3)$-cliques.

    Finally, after creating the last grid column, we create the outer cuts between the last id gadgets and the first id gadgets of each color gadget, by creating an independent set of size $n$ and joining it with all labels $\ell^0_c$ for $c\in C$ for each color $C$. This concludes the construction of $\hat{G}$.

    It is not hard to see that the resulting graph is exactly the graph $\hat{G}$. In particular, all colored vertices are assigned edges defined by their colors, and not by their ids, and therefore, assigning the labels accordingly yields the correct adjacencies. Since we only use $\Oh(k)$ labels, the width of the constructed linear clique expression is $\Oh(k)$. Clearly, such a linear clique expression can be constructed in polynomial time given the instance $H$.
\end{proof}

\begin{proof}[Proof of \cref{theo:lb}]
    The $W[1]$-hardness follows from \cref{cor:binomial-multi-colored-clique}, since the reduction presented above fulfills all the requirements of a parameterized reduction.
    For the latter part, assume that there exists an algorithm solving $d$-\ClqPartp in time $n^{o(k^{d-1})}$ on graphs of linear clique-width $k$ when given together with a linear $k$-clique expressions. We show that this would contradict ETH by providing an algorithm solving $(d-1)$-\textsc{Binomial Multi-Colored Clique} in time $n^{o\big(\binom{k}{d}\big)}$.

    Given an instance $H$ of $(d-1)$-\textsc{Binomial Multi-Colored Clique}, we construct the graph $\hat{G}$ together with a linear $k$-clique-expression of $\hat{G}$ of width $\Oh(k)$ as described above. By \cref{lem:lb-linear-cw} this can be done in polynomial time. It holds by \cref{lem:lb-prt-if-clq} and \cref{lem:lb-clq-if-prt} that $H$ has a multi-colored clique if and only if $\hat{G}$ has a $d$-clique partitioning. Hence, we can solve the instance $H$ by running the assumed algorithm on the instance $\hat{G}$. Since the size of $\hat{G}$ is bounded polynomially, this algorithm solves the instance $H$ in time $n^{o(k^{d-1})}$, contradicting ETH.
\end{proof}

\section{Conclusion}\label{section:conclusion}

We have presented an algorithm that for each $d\geq 3$ solves $d$-\ClqPackp in $n^{\Oh(k^{d-1})}$ time. We complemented this by ruling out algorithms running in $n^{o(k^{d-1})}$ time assuming ETH, even for the special case of $d$-\ClqPartp.

It would be interesting to further explore the packing problems for (induced) $H$-subgraphs for fixed graph $H$ relative to clique-width. For the induced case, we expect the same tight bound. Some preliminary results suggest that (unsurprisingly) things get complicated for packing $H$-subgraphs, depending on $H$. Are there other natural problems with optimal XP-algorithms with similar running times relative to clique-width?

\bibliography{ref}

\end{document}